\documentclass[smallextended]{svjour3}     
\smartqed  
\usepackage{graphicx}
\usepackage{mathptmx}      
\usepackage{amsmath}
\usepackage{amsfonts}
\usepackage{mathtools}
\usepackage{bbm}
\usepackage{paralist}

\providecommand{\MoveEqLeft}{}

\newlength\PullBackLength
\renewcommand\MoveEqLeft[1][2]{%
  \setlength{\global\PullBackLength}{#1em}%
  \kern\PullBackLength%
  & 
  \kern-\PullBackLength}

\let\epsilon\varepsilon

\let\phi\varphi

\let\theta\vartheta

\DeclareMathOperator{\ad}{ad}

\DeclareMathOperator*{\slim}{s-\!\lim}
\DeclareMathOperator*{\wlim}{w-\!\lim}
\DeclareMathOperator{\w}{w-\!}
\DeclareMathOperator{\uP}{P}
\DeclareMathOperator{\Ran}{Ran}
\newcommand{\ri}{\mathrm{i}}
\newcommand{\pp}{\mathrm{pp}}
\providecommand{\C}{\mathbb{C}} \providecommand{\N}{\mathbb{N}}
\providecommand{\R}{\mathbb{R}} \providecommand{\cH}{\mathcal{H}}
\providecommand{\cD}{\mathcal{D}} \providecommand{\cB}{\mathcal{B}}
 
 \providecommand{\cS}{\mathcal{S}}
\providecommand{\cK}{\mathcal{K}} \providecommand{\cO}{\mathcal{O}}
\newcommand{\inner}[3][]{#1\langle #2,#3 #1\rangle}
\newcommand{\bra}[2][]{#1\langle #2\rvert}
\newcommand{\ket}[2][]{#1\lvert #2\rangle}
\newcommand{\bbbone}{\mathbbm{1}}

\newtheorem{condition}{Condition}

\journalname{Letters in Mathematical Physics}

\providecommand{\kdj}{\delta_j}

\DeclareMathOperator{\adjungeret}{ad}
\DeclareMathOperator{\ud}{d\!}
\providecommand{\ad}[2]{\adjungeret_{#1}^{#2}}

\providecommand{\abs}[2][]{#1\lvert#2#1\rvert}
\providecommand{\norm}[2][]{#1\lVert#2#1\rVert}
\providecommand{\jnorm}[2][]{#1\langle#2#1\rangle}
\providecommand{\inv}{^{-1}}
\providecommand{\ext}{^{\mathrm{ext}}}

\newcommand{\bd}{\mathbf{d}}
\newcommand{\bnd}{\mathrm{bd}}
\newcommand{\bD}{\mathbf{D}}

\DeclareMathOperator{\hc}{h.c.}

\begin{document}

\title{Asymptotic Completeness in Quantum Field Theory: Translation
  Invariant Nelson Type Models Restricted to the Vacuum and
  One-Particle Sectors}


\titlerunning{Translation Invariant Nelson Type Models Restricted to
  the Vacuum and One-Particle Sectors} 

\author{Christian G\'erard \and Jacob Schach M\o{}ller \and Morten
  Grud Rasmussen}

\authorrunning{G\'erard, M\o{}ller, Rasmussen} 

\institute{Christian G{\'e}rard 
  \at D\'epartement de Math\'ematiques, Universit\'e Paris-Sud,
  B\^at. 425, F-91405 Orsay Cedex, France\\
  \email{christian.gerard@math.u-psud.fr}
  \and
  Jacob Schach M\o{}ller 
  \and Morten Grud Rasmussen 
  \at Department of Mathematical Sciences, Aarhus University,
  Bldg. 1530, DK-8000 Aarhus C, Denmark\\
  \email{jacob@imf.au.dk}\and
  \email{mgr@imf.au.dk}
}

\date{Received: date / Accepted: date}

\maketitle

\begin{abstract}
  Time-dependent scattering theory for a large class of translation
  invariant models, including the Nelson and Polaron models,
  restricted to the vacuum and one-particle sectors is
  studied. We formulate and prove asymptotic completeness for these models. The
  translation invariance imply that the Hamiltonians considered are fibered with
  respect to the total momentum. On the way to asymptotic completeness
  we determine the spectral structure of the fiber Hamiltonians,
  establish a Mourre estimate and derive a geometric asymptotic
  completeness statement as an intermediate step.\newline

  \noindent{\bfseries Keywords:} quantum field theory, time-dependent
  scattering theory, asymptotic completeness, translation
  invariance\newline {\bfseries Mathematics Subject Classification
    (2010):} 81Q10, 47A40, 81T10, 81U30
\end{abstract}

\tableofcontents

\section{Introduction and motivation}\label{sec:intro}

In this paper, we study the spectral and scattering theory of a class
of Hamiltonians that arise when one restricts e.g. the Nelson or
Polaron model to the subspace of at most one field particle. As our
results are valid for both models, we will use the term ``field
particles'' rather than photons or phonons, and in the same spirit, we
will use the term ``matter particle'' rather than electron or
positron.

In \cite{MGRJSM}, two of the authors prove a Mourre estimate and $C^2$
regularity for the full model, with respect to a suitably chosen
conjugate operator. The estimate holds in the part of the
energy-momentum spectrum lying between the bottom of the essential
energy-momentum spectrum and either the two-body threshold, if there
are no exited isolated mass shells, or the one-body threshold
pertaining to the first exited isolated mass shell, if it exists. This
is a natural first step for scattering theory. As the full model in
that energy-momentum regime is expected to resemble the model with at
most one field particle in many aspects, the scattering theory of the
cut-off model is of obvious interest. We note that in \cite{GJY}, the
spectral and scattering theory of the massless Nelson model is
studied. The stationary methods used there to prove
asymptotic completeness would to some
extend also work on the class of models considered here. However, the
scattering theory in \cite{GJY} is obtained via a Kato-Birman argument
which one cannot hope to work on the full model. The present paper should be seen as a
test case for the application of the time-dependent methods from
\cite{DGe} to translation invariant models.

In recent years a lot of effort was put into investigating the
spectral and scattering theory of various models of quantum field
theory (see among many other papers \cite{Am}, \cite{AMiZ},
\cite{DGe}, \cite{FGrSch}, \cite{FGrS}, \cite{Ge}, \cite{P}, \cite{Sp}
and references therein). Substantial progress was made by applying
methods originally developed in the study of $N$-particle
Schr\"o{}dinger operators namely the Mourre positive commutator method
and the method of propagation observables to study the behavior of the
unitary group $e^{-itH}$ for large times. Up to now, the most complete
results on the scattering theory for these models have only been
available for models where the translation invariance is broken
\cite{Am}, \cite{DGe}, \cite{Ge}, \cite{P}, \cite{Sp}, or for small
coupling constants \cite{FGrSch}. In fact the only asymptotic
completeness result valid for arbitrary coupling strength, in
time-dependent scattering theory of translation invariant models known
to us are variations of the $N$-body problem, where the dispersion
relations are of the non-relativistic form $\frac{p^2}{M}$. Our
results hold for a large class of dispersion relations, including a
combination of the relativistic and non-relativistic choices.

In order to appreciate the difficulties associated with proving
asymptotic completeness for translation invariant models of QFT, we
explain the structure of scattering channels.  If a system starts in a
scattering state at total momentum $\xi$ and energy $E$, it will emit
field particles with momenta $k_1,\dots,k_n$ until the remaining
interacting system reaches a total momentum $\xi'$ and an eigenvalue
$E'(\xi')$ for the Hamiltonian at total momentum $\xi'$.  In order to
conserve energy and momentum we must have $\xi = \xi'+k_1+\cdots +
k_n$ and $E= E'(\xi') + \omega(k_1)+\cdots +\omega(k_n)$, where
$\omega$ is the dispersion relation for the field.

That is, the scattering channels are labeled by bound states at
momenta $\xi'$ and the number of emitted field particles $n$, under
the constraint of conservation of energy and total momentum.  The
resulting bound particle will not be at rest but rather move according
to a dispersion relation which is in fact the eigenvalue band, or mass
shell, to which it belongs.  This band may a priori be an isolated
mass shell or an embedded one.  If one wants to capture the behaviour
of scattering states through a Mourre estimate, then one needs to
build into a conjugate operator the dynamics of all the mass shells
that appear in the available channels. This is a difficult task.  The
thresholds at total momentum $\xi$ are energies $E$ that has a
scattering channel with the property that the bound state and the
emitted field particles do not separate over time.

When introducing a number cutoff in the model, one simplifies the
situation in that the scattering channels are now labeled by bound
states of Hamiltonians with strictly fewer field particles. In
particular in our case, we can label the scattering channels by mass
shells of the Hamiltonian on the vacuum sector, which are easily
understood. Indeed, there is in fact only one mass shell and it is
identical to the matter dispersion relation $\Omega$.

Finally, we will briefly outline the contents of this paper. In
Section~\ref{sec:model} we introduce the model in details and state
our main result on asymptotic completeness. In
Section~\ref{sec:specanal} we briefly go through the spectral theory
for the fiber Hamiltonians, in particular we prove an HVZ theorem, a
Mourre estimate, absence of singular continuous spectrum and a
semi-continuity of the Mourre estimate. In
Section~\ref{sec:propest} we prove the following propagation
estimates: A large velocity estimate, a phase-space propagation
estimate, an improved phase-space propagation estimate and a minimal
velocity estimate. These form the technical foundation for
Section~\ref{sec:asympobs}, where we begin by introducing a key asymptotic
observable, which gives rise to spaces of asymptotically bound resp. free particles.
Finally we construct wave operators and prove asymptotic completeness
via a so-called geometric asymptotic completeness result.

\section{The model and the result}\label{sec:model}

The Hilbert space for the Hamiltonian is
\begin{equation*}
\cH=L^2(\R^\nu,\ud y)\otimes (\C\oplus L^2(\R^\nu,\ud
x))=L^2(\R^\nu,\ud y)\oplus L^2(\R^{2\nu},\ud x\ud y),
\end{equation*}
where $\nu\in\N$. We write $D_x=-\ri \nabla_x$, $D_y=-\ri \nabla_y$ for the
respective momentum operators. The Hamiltonian we wish to study the
spectral and scattering theory of is given by
\begin{equation*}
  H=H_0+V=\begin{pmatrix}
    \Omega(D_y)&0\\0&\Omega(D_y)+\omega(D_x)
  \end{pmatrix}+
  \begin{pmatrix}
    0&v^*\\v&0
  \end{pmatrix},
\end{equation*}
where 
\begin{equation*}
  (vu_0)(x,y)=\rho(x-y)u_0(y)\quad\text{and}\quad(v^*u_1)(x)=\int\rho(x-y)u_1(x,y)dy
\end{equation*}
for some $\rho\in L^2(\R^\nu)$. Here $\Omega$ is the dispersion
relation for the matter particle, $\omega$ the dispersion relation for
the field particles and $\rho$ a coupling function. One may view it as
the translation invariant Nelson or Polaron model restricted to the
subspace with at most one field particle, depending on the choice of
dispersion relations.

The coupling function will be assumed to satisfy a short-range
condition which implies a UV-cutoff (see Condition~\ref{cond:rho}). We
work with more general dispersion relations $\omega$ and $\Omega$ than
$\omega(k)=\sqrt{k^2+m^2}$ or $\omega(k)=\omega_0>0$ and
$\Omega(\eta)=\eta^2/2M$ respectively (see Conditions~\ref{cond:Omega}
and \ref{cond:omega} for details).  As the infrared problem is not
present in this model due to the finite number of field particles, the
mass of the field particle is not important. However, the singular
behavior of the dispersion relation $\omega(k)=\abs{k}$ at $k=0$ makes
this choice fall outside of what can be handled in this treatment,
although it seems likely that one with minor adjustments may include
this case in the same framework. For a treatment of the case where
$\Omega(\eta)=\frac{1}{2}\eta^2$ and $\omega(k)=\abs{k}$, see
\cite{GJY}.

The operator $H$ commutes with the operator of total momentum,
\begin{equation*}
  \uP\!=\begin{pmatrix}
    D_y&0\\0&D_x+D_y
  \end{pmatrix},
\end{equation*}
and hence $H$ is fibered, $H=U\inv\int_{\R^\nu}^\oplus
H(P)\ud PU$, where 
\begin{equation*}
  U(u_0,u_1)(x,y)=(u_0(y),u_1(y,x+y))
\end{equation*}
and
\begin{equation*}
  H(P)=H_0(P)+\tilde V=\begin{pmatrix}\Omega(P)&0\\
    0&\Omega(P-D_x)+\omega(D_x)
  \end{pmatrix}+
  \begin{pmatrix}
    0&\bra{\rho}\\
\ket{\rho}&0
  \end{pmatrix},
\end{equation*}
where $\bra{\cdot}$ and $\ket{\cdot}$ denote the Dirac brackets. The
fiber Hamiltonians are operators on the Hilbert space $\cK=\C\oplus
L^2(\R^\nu)$.

The precise assumptions on $\Omega$, $\omega$ and $\rho$ are given
below. We adopt the standard notation
$\jnorm{x}=(1+{x}^2)^{\frac{1}{2}}$.

\begin{condition}[Matter particle dispersion relation]\label{cond:Omega}
  Let $\Omega\in C^\infty(\R^\nu)$ be a non-ne\-ga\-ti\-ve, real-analytic
  and rotation invariant\footnote{By rotation invariance of a function
    $f$ we mean that $f(\eta)=f(O\eta)$ a.e. for any $O\in O(\nu)$
    where $O(\nu)$ denotes the $\nu$-dimensional orthogonal group.}
  function. There exists $s_\Omega\in[0,2]$ such that $\Omega$
  satisfies:
  \begin{enumerate}[(i)]
  \item There is a $C>0$ such that $\Omega(\eta)\ge
    C\inv\jnorm{\eta}^{s_\Omega}-C$.
  \item For any multi-index $\alpha$ there is a $C_\alpha>0$
    such that $\abs{\partial^\alpha\Omega(\eta)}\le
    C_\alpha\jnorm{\eta}^{s_\Omega-\abs{\alpha}}$.
  \end{enumerate}
\end{condition}

Note that this assumption is satisfied by the standard
non-relativistic and relativistic choices,
$\Omega(\eta)=\frac{\eta^2}{2M}$ and $\Omega(\eta)=\sqrt{\eta^2+M^2}$.

\begin{condition}[Field particle dispersion relation]\label{cond:omega}
  Let $\omega\in C^\infty(\R^\nu)$ be non-ne\-ga\-ti\-ve, real-analytic,
  rotation invariant and satisfy:
  \begin{enumerate}[(i)]
  \item For any multi-index $\alpha$ with $\abs{\alpha}\ge1$, we have
    $\sup_{k\in\R^\nu}\abs{\partial^\alpha\omega(k)}<\infty.$
  \item If $s_\Omega=0$, then $\omega(k)\to\infty$ as
    $\abs{k}\to\infty$.
  \end{enumerate}
\end{condition} 

This is satisfied e.g. for $\omega(k)=\sqrt{k^2+m^2}$,
$m\neq0$, and also for the semi-relativistic and non-relativistic
large polaron models, where $\omega(k)=\omega_0$.

\begin{condition}[Coupling function]\label{cond:rho}
  Let $\rho\in L^2(\R^\nu)$ be rotation invariant and satisfy that
  \begin{enumerate}[(i)]
  \item $\hat\rho\in C^2(\R^\nu)$.
  \item
    $\jnorm{\cdot}\abs{\nabla\hat\rho},\jnorm{\cdot}\norm{\nabla^2\hat\rho}\in
    L^2(\R^\nu)$.
  \item\label{cond:item:SR} There exist constants $C,\mu>0$ such that $\abs{\rho(x)}\le
    C\jnorm{x}^{-1-\frac{\nu}{2}-\mu}$.
  \end{enumerate}
\end{condition}

Condition~\ref{cond:rho}~\eqref{cond:item:SR} is the so-called
short-range condition. Note that it implies that for $J\in
C^\infty(\R^\nu)$ with support away from $0$, we have
\begin{equation}
  \label{eq:SR}
  \norm{J(\tfrac{x}{t})\rho}=O(t^{-1-\mu}).
\end{equation}
For the rest of this paper,
Conditions~\ref{cond:Omega}, \ref{cond:omega} and \ref{cond:rho} will
tacitly be assumed to be fulfilled, and under this assumption, our
main result will be the following

\begin{theorem}[Asymptotic completeness]\label{thm:main}
  The wave operator 
  \begin{equation*}
    W^+=\slim_{t\to\infty}e^{\ri tH}e^{-\ri tH_0}P^+(H_0)
  \end{equation*}
  exists, where $P^+(H_0)$ is the projection onto $\{0\}\oplus
  L^2(\R^{2\nu})$, and the system is asymptotically complete: 
  \begin{equation*}
    \Ran W^+=\cH_\bnd^\perp, 
  \end{equation*}
  where $\cH_\bnd=U\inv\int_{\R^\nu}^\oplus \bbbone_{\mathrm{pp}}(H(P))\ud PU\cH$.
\end{theorem}

\begin{remark}That $P\mapsto\bbbone_{\mathrm{pp}}(H(P))$ is weakly --
  and hence strongly -- measurable follows from an application of the
  RAGE theorem, \cite[Theorem~5.8]{CFKS}, see the proof of
  \cite[Theorem~9.4]{CFKS} for details.
\end{remark}

\section{Spectral analysis}

\label{sec:specanal}
We begin by recalling the following well-known properties of the
fibered Hamiltonian. The Hamiltonian $H_0(P)$ is essentially
self-adjoint on $\C\oplus C_0^\infty(\R^\nu)$ and the domain
$\cD=\cD(H_0(P))$ is independent of $P$. As $\tilde V$ is bounded, the
Kato-Rellich theorem implies that the same is true for $H(P)$ and that
$\cD(H(P))=\cD$.

The following threshold set will play an important role in our
analysis:
\begin{equation*}
  \theta(P)=\bigl\{\lambda\in\R\,\big|\,\exists
  k\in\R^\nu:\lambda=\Sigma(P-k)+\omega(k),
  \nabla\Omega(P-k)=\nabla\omega(k)\bigr\}.
\end{equation*}
The energies $E$ comprising $\theta(P)$ are those for which
interacting states sharply
localized at energy $E$, may decay into a boson and a free particle
that do not break up over time. That is, emitted bosons, at threshold
energies,  may not escape the interaction region. Clearly $\theta(P)$ only depends on $P$
up to rotations. It is essential for our analysis that $\theta(P)$ 
is a closed set of measure zero, in fact it is locally finite. 
This follows from real analyticity and rotation invariance of the functions $\omega$ and
$\Omega$. A similar argument played a role in \cite{MProc}.

The following results, Theorems~\ref{thm:C2} to \ref{thm:Mourre},
correspond to completely analogous statements for the full model, see
\cite{MGRJSM}.  When $H$ is of class $C^1(A)$, we denote by $[H,iA]^\circ$ the unique extension of the commutator
form $[H,iA]$ defined on $\cD(A)\cap\cD(H)$ to an element of $\cB(\cD(H);\cD(H)^*)$. See
Appendix~\ref{app:B} for the definition of the $C^k(A)$, $k\in\mathbb{N}$, classes.

\begin{theorem}\label{thm:C2}Assume that the vector field $v_P\in C^\infty(\R^\nu;\R^\nu)$
  satisfies that for any multi-index $\alpha$,
  $\abs{\alpha}\in\{0,1,2\}$, there is a constant $C_\alpha>0$ such
  that $\abs{\partial^\alpha v_P(\eta)}\le
  C_\alpha\jnorm{\eta}^{1-\abs{\alpha}}$.  Then the operator
  $a_P=\frac{1}{2}(v_P(D_x)\cdot x+x\cdot v_P(D_x))$ is essentially
  self-adjoint on the Schwarz space $\cS$ and $H(P)$ is of class
  $C^2(A_P)$, where $A_P=(\begin{smallmatrix}0&0\\0&a_P
  \end{smallmatrix})$ is self-adjoint on $\cD(A_P)$. The first
  commutator is given by 
  \begin{equation*}
    [H(P),\ri A_P]^\circ=
    \begin{pmatrix}
      0&\bigl\langle \ri a_P\rho \bigr|\\
      \left|\ri a_P\rho \right\rangle&v_P(D_x)\cdot\nabla(\omega(D_x)+\Omega(P-D_x))
    \end{pmatrix}
  \end{equation*}
  as a form on $\cD$.
\end{theorem}
This can be seen either by direct computations or by following
\cite{MGRJSM}.

We now introduce the extended space $\cK\ext=\cK\oplus L^2(\R^\nu)$ to
be able to make a geometric partition of unity in configuration
space. The partition of unity is similar to what is done in the
analysis of the $N$-body Schr\"o{}dinger operator (see
e.g. \cite{DGeBook}) and in complete analogy with what is done in
e.g. \cite{DGe} and \cite{MAHP}. The partition of unity used here may
actually be seen as the partition of unity introduced in \cite{DGe}
restricted to the subspace with at most $1$ field particle.

Let $j_0,j_\infty\in C^\infty(\R^\nu)$ be real, non-negative functions
satisfying $j_0=1$ on $\{x\,|\,\abs{x}\le 1\}$, $j_0=0$ on
$\{x\,|\,\abs{x}>2\}$ and $j_0^2+j_\infty^2=1$. We now define
\begin{align*}
  j^R\colon\cK&\to\cK\ext\\
  j^R(v_0,v_1)&=(v_0,j_0(\tfrac{\cdot}{R})v_1)\oplus(j_\infty(\tfrac{\cdot}{R})v_1).
\end{align*}
Clearly, $j^R$ is isometric. 

We introduce two self-adjoint operators, the extended Hamiltonian,
$H\ext(P)$, and the extended conjugate operator, $A_P\ext$, acting in
$\cK\ext$,
\begin{align*}
  H\ext(P)&=H(P)\oplus F_P(D_x)\text{ and}\\
  A_P\ext&=A_P\oplus a_P,
\end{align*}
where $F_P(D_x)=\omega(D_x)+\Omega(P-D_x)$, with the obvious domains
denoted by $\cD\ext$ and $\cD(A_p\ext)$. The extended Hamiltonian
describes an interacting system together with a free field
particle. It is easy to see that Theorem~\ref{thm:C2} holds true with
$H(P)$ and $A_P$ replaced by $H\ext(P)$ and $A_P\ext$, respectively,
and the commutator equal to
\begin{equation*}
  [H\ext(P),\ri A_p\ext]^\circ=[H(P),\ri A_P]^\circ\oplus \bigl(v_P(D_x)\cdot(\nabla\omega(D_x)-\nabla\Omega(P-D_x))\bigr).
\end{equation*}
We have the following localisation error when applying $j^R$.

\begin{lemma}\label{lem:locerror} Let $f\in C_0^\infty(\R)$. Then
  \begin{align*}
    \MoveEqLeft j^Rf(H(P))=f(H\ext(P))j^R+o_R(1)\quad\text{and}\\
    \MoveEqLeft j^Rf(H(P))[H(P),\ri A_P]^\circ
    f(H(P))\\
    &=f(H\ext(P))[H\ext(P),\ri A_P\ext]^\circ f(H\ext(P))j^R+o_R(1),
  \end{align*} 
  for $R\to\infty$.
\end{lemma}
This can be seen either by a direct computation or by applying
\cite[Corollary~5.3]{MGRJSM}. The following two results, an HVZ
theorem and a Mourre estimate, are now almost immediate.

\begin{theorem}\label{thm:HVZ}
  The spectrum of $H(P)$ below
  $\Sigma_{\mathrm{ess}}(P)=\inf_{k\in\R^\nu}\{\Omega(P-k)+\omega(k)\}$
  consists at most of eigenvalues of finite multiplicity and can only
  accumulate at $\Sigma_{\mathrm{ess}}(P)$. The essential spectrum is
  given by
  $\sigma_{\mathrm{ess}}(H(P))=\left[\Sigma_\mathrm{ess}(P),\infty\right)$.
\end{theorem}

\begin{proof}
  Using Lemma~\ref{lem:locerror} for an $f\in C_0^\infty(\R)$
  supported in $\left(-\infty,\Sigma_\mathrm{ess}(P)\right)$ and
  letting $R$ tend to infinity shows that $f(H(P))$ is compact. This
  proves the first part.

  To prove the last part, let
  $\lambda\in\left[\Sigma_\mathrm{ess}(P),\infty\right)$ and note that
  there exists a $k_0\in\R^\nu$ such that
  $\lambda=\Omega(P-k_0)+\omega(k_0)$. Now choose
  $u_n=(0,u_{1n})\in\C\oplus L^2(\R^\nu)$ with $\hat
  u_{1n}(\cdot)=n^{\frac{\nu}{2}}f(n(\cdot-k_0))$ for some $f\in
  C_0^\infty(\R^\nu)$ with $f\ge0$ and $f(0)=1$. One may now check
  that $u_n$ is a Weyl sequence for the energy $\lambda$.
\end{proof}

\begin{theorem}\label{thm:Mourre}
  Assume that $\lambda\not\in\theta(P)$. Let $A_P$ be given as in
  Theorem~\ref{thm:C2} with
  $v_P(D_x)=\nabla\omega(D_x)-\nabla\Omega(P-D_x)$. Then there exist
  constants $\kappa,c>0$ and a compact operator $K$ such that
  \begin{align*}
    E_{\lambda,\kappa}(H(P))[H(P),\ri A_P]^\circ
    E_{\lambda,\kappa}(H(P))\ge cE_{\lambda,\kappa}(H(P))+K,
  \end{align*}
  where $E_{\lambda,\kappa}$ denotes the characteristic function of
  the interval $[\lambda-\kappa,\lambda+\kappa]$.
\end{theorem}

\begin{proof}
  We may find a $\kappa$ such that
  $[\lambda-2\kappa,\lambda+2\kappa]\cap\theta(P)=\emptyset$. Choose
  $f\in C_0^\infty(\R)$ with support in
  $[\lambda-2\kappa,\lambda+2\kappa]$ and equal to $1$ on
  $[\lambda-\kappa,\lambda+\kappa]$. Note that
  \begin{align*}
    \MoveEqLeft f(H(P))[H(P),\ri A_P]^\circ
    f(H(P))\\
&={j^R}^*j^Rf(H(P))[H(P),\ri A_P]^\circ f(H(P))\\
    &={j^R}^*f(H\ext(P))[H\ext(P),\ri A_P\ext]^\circ
    f(H\ext(P))j^R+o_R(1),
  \end{align*}
  by Lemma~\ref{lem:locerror}. Note that
  \begin{align}
    \MoveEqLeft f(H\ext(P))[H\ext(P),\ri A_P\ext]^\circ
    f(H\ext(P))j^R\nonumber\\
    \MoveEqLeft[1]=f(H(P))[H(P),\ri A_P]^\circ f(H(P))\begin{pmatrix}1&0\\0&j_0(\tfrac{\cdot}{R})
    \end{pmatrix}\label{eq:MourreBK}\\
    &\oplus
    f(F_P(D_x))\abs[\big]{\nabla\omega(D_x)-\nabla\Omega(P-D_x)}^2f(F_P(D_x))j_\infty(\tfrac{\cdot}{R}).\nonumber
  \end{align}
  Taking the support of $f$ into account, one finds that
  \begin{equation*}
    f(F_P(D_x))\abs[\big]{\nabla\omega(D_x)-\nabla\Omega(P-D_x)}^2f(F_P(D_x))\ge
    2c f^2(F_P(D_x))
  \end{equation*}
  for some positive constant $c>0$. It is easy to see that
  \[
    K(R)=f(H(P))\begin{pmatrix}1&0\\0&j_0(\frac{\cdot}{R})\end{pmatrix}
  \] 
  is compact.  Let $g\in C_0^\infty(\R)$ equal $1$
  on the support of $f$. Then 
  \begin{equation*}
    B=f(H(P))[H(P),\ri A_P]^\circ g(H(P))
  \end{equation*}
  is bounded and \eqref{eq:MourreBK} equals $BK(R)$. Hence by Lemma~\ref{lem:locerror}
  \begin{align*}
    \MoveEqLeft f(H(P))[H(P),\ri A_P]^\circ f(H(P))\\
    \MoveEqLeft[1]\ge{j^R}^*2cf^2(H(P))\begin{pmatrix}1&0\\0&j_0(\frac{\cdot}{R})
  \end{pmatrix}\oplus 2cf^2(F_P(D_x))j_\infty(\tfrac{\cdot}{R})\\
  &+{j^R}^*\bigl(B-2cf(H(P))\bigr)K(R)\oplus 0+o_R(1)\\
  &=2cf^2(H(P))+K_R+o_R(1),
  \end{align*}
  for some compact operator $K_R$ depending on $R$. One may now choose
  $R$ so large that $\norm{o_R(1)}\le c$ and sandwich the inequality
  with $E_{\lambda,\kappa}(H(P))$ on both sides to arrive at the
  desired result.
\end{proof}

We infer the following corollary of Theorems~\ref{thm:C2} and
\ref{thm:Mourre} by standard arguments of regular Mourre theory.

\begin{corollary}
  The essential spectrum of the fiber Hamiltonians is non-singular:
  \begin{equation*}
    \sigma_{\mathrm{sing}}(H(P))=\emptyset.
  \end{equation*}
\end{corollary}

\begin{theorem}\label{thm:lowsemcontMourre}
  Let $(P_0,\lambda_0)\in\R^{\nu+1}$. Assume that
  $\lambda_0\not\in\theta(P_0)\cup\sigma_{\mathrm{pp}}(P_0)$. Then
  there exists a constant $C>0$, a neighbourhood $\cO$ of $P_0$ and a
  function $f\in C_0^\infty(\R)$ with $f=1$ in a neighbourhood of
  $\lambda_0$ such that for all $P\in \cO$,
  \begin{equation*}
    f(H(P))[H(P),\ri A_{P_0}]^\circ f(H(P))\ge Cf^2(H(P))
  \end{equation*}
  where $A_{P_0}$ is given as in Theorem~\ref{thm:Mourre}.
\end{theorem}

\begin{proof}
  We begin by noting that the object $[H(P),\ri A_{P_0}]^\circ$ is
  well-defined by Theorem~\ref{thm:C2}. By standard arguments using
  the fact that $\lambda_0\not\in\sigma_{\mathrm{pp}}(P_0)$ and Theorem~\ref{thm:Mourre}, there exist a function $\tilde f\in C_0^\infty(\R)$ and a constant $\tilde C$ such that
  \begin{equation*}
    \tilde f(H(P_0))[H(P_0),\ri A_{P_0}]^\circ\tilde f(H(P_0))\ge\tilde C\tilde f^2(H(P_0)),
  \end{equation*}
  with $\tilde f=1$ on a neighbourhood of $\lambda_0$. It is easy to
  see that ${(H_0(0)-\ri )(H(P)-z)\inv}$ and $(H_0(0)-\ri
  )\inv[H(P),\ri A_{P_0}]^\circ(H_0(0)-\ri )\inv$ are norm continuous
  as functions of $P$, and hence it follows by an application of the
  functional calculus of almost analytic extensions that $\tilde
  f^2(H(P))$ and $\tilde f(H(P))[H(P),\ri A_{P_0}]^\circ\tilde
  f(H(P))$ are norm continuous as functions of $P$.

  Let $\cO\ni P_0$ be a neighbourhood such that 
  \begin{gather*}
    \norm{\tilde f^2(H(P))-\tilde f^2(H(P_0))}\le \frac{\tilde C}{3}\quad\text{and}\\
    \norm{\tilde
  f(H(P))[H(P),\ri A_{P_0}]^\circ\tilde f(H(P))-\tilde
  f(H(P_0))[H(P_0),\ri A_{P_0}]^\circ\tilde f(H(P_0))}\le\frac{\tilde C}{3}
  \end{gather*}
  for all $P\in \cO$. Then 
  \begin{equation}\label{eq:lowsemcontproof}
     \tilde f(H(P))[H(P),\ri A_{P_0}]^\circ\tilde f(H(P))
     \ge-\frac{2\tilde C}{3}I+\tilde C\tilde f^2(H(P)).
  \end{equation}
  Choose now $C=\frac{\tilde C}{3}$ and $f\in C_0^\infty(\R)$ such
  that $f=1$ on a neighbourhood of $\lambda_0$ and $f=f\tilde f$. The
  result is then obtained by multiplying \eqref{eq:lowsemcontproof}
  from both sides with $f(H(P))$.
\end{proof}

\section{Propagation estimates}\label{sec:propest}

We will write $\bD=[H,\ri\,\cdot\,]+\frac{\ud}{\ud t}$ and
$\bd_0=[\Omega(D_x+D_y)+\omega(D_x),\ri\,\cdot\,]+\frac{\ud}{\ud t}$
for the Heisenberg derivatives. The following abbreviation will be
used to ease the notation:
  \begin{equation}
    [B]:=\begin{pmatrix}0&0\\0&B\end{pmatrix}.\label{eq:bracketnotation}
  \end{equation}

\begin{theorem}[Large velocity estimate]\label{thm:largevel}
  Let $\chi\in C_0^\infty(\R)$. There exists a constant $C_1$ such
  that for $R'>R>C_1$, one has
  \begin{equation*}
    \int_1^\infty\norm[\big]{\bigl[\bbbone_{[R,R']}\bigl(\tfrac{\abs{x-y}}{t}\bigr)
    \bigr]e^{-\ri tH}\chi(H)u}^2\frac{\ud t}{t}\le C\norm{u}^2
  \end{equation*}
\end{theorem}

\begin{proof}
  Let $C_1$ be a constant to be specified later and $R'>R>C_1$. Let
  $F\in C^\infty(\R)$ equal $0$ near the origin and $1$ near
  infinity such that $F'(s)\ge c\bbbone_{[R,R']}(s)$ for some positive
  constant $c>0$. Let
  \begin{align*}
    \Phi(t)&=-\chi(H)[F(\tfrac{\abs{x-y}}{t})]\chi(H),\\
    b(t)&=-\bd_0F(\tfrac{\abs{x-y}}{t}).
  \end{align*}
  By using e.g.\ Theorem~\ref{thm:MGR} or
  pseudo-differential calculus one sees that
  \begin{equation*}
    b(t)=\frac{1}{t}\bigl(\tfrac{\abs{x-y}}{t}-(\nabla\Omega(D_y)-\nabla\omega(D_x))\cdot\tfrac{x-y}{\abs{x-y}}\bigr)F'\bigl(\tfrac{\abs{x-y}}{t}\bigr)+O(t^{-2}).
  \end{equation*}
  Hence for any $\tilde\chi\in C_0^\infty(\R)$ such that $\chi=\chi\tilde\chi$ one finds that
  \begin{align*}
    \MoveEqLeft-\chi(H)[b(t)]\chi(H)\\
    &=\frac{1}{t}\chi(H)\Bigl(\tfrac{\abs{x-y}}{t}-(\nabla\Omega(D_y)-\nabla\omega(D_x))\cdot\tfrac{x-y}{\abs{x-y}}\Bigr)F'\bigl(\tfrac{\abs{x-y}}{t}\bigr)\chi(H)+O(t^{-2})\\
    &=\frac{1}{t}\chi(H)\Bigl(\tfrac{\abs{x-y}}{t}-\tilde\chi(H)(\nabla\Omega(D_y)-\nabla\omega(D_x))\cdot\tfrac{x-y}{\abs{x-y}}\Bigr)\bbbone_{[C_1,\infty)}\bigl(\tfrac{\abs{x-y}}{t}\bigr)\\
    &\qquad\times F'\bigl(\tfrac{\abs{x-y}}{t}\bigr)\chi(H)+O(t^{-2})\\
    &\ge \frac{C_0}{t}\chi(H)F'\bigl(\tfrac{\abs{x-y}}{t}\bigr)\chi+O(t^{-2})
  \end{align*} for some $C_0>0$ if one chooses
  $C_1>\norm{\tilde\chi(H)(\nabla\Omega(D_y)-\nabla\omega(D_x))\tfrac{x-y}{\abs{x-y}}}$.
  
  It follows from Condition~\ref{cond:rho}~\eqref{cond:item:SR} that
  \begin{equation*}
    [V,\ri [F(\tfrac{\abs{x-y}}{t})]]=O(t^{-1-\mu}),
  \end{equation*}
  cf. \eqref{eq:SR}.  Putting this together, we get
  \begin{align*}
    \bD\Phi(t)\ge\tfrac{C_0}{t}\chi(H)[F'(\tfrac{\abs{x-y}}{t})]\chi(H)+O(t^{-1-\mu}),
  \end{align*}
  which combined with Lemma~\ref{lem:propest} implies the result.
\end{proof}

\begin{theorem}[Phase-space propagation estimate]\label{thm:phase-space}
  Let $\chi\in C_0^\infty(\R)$, $0<c_0<c_1$. Write
  \begin{align*}
    \MoveEqLeft[1]\Theta_{[c_0,c_1]}(t)=\\
    &\bigl[\inner[\big]{\tfrac{x-y}{t}-\nabla\omega(D_x)+\nabla\Omega(D_y)}{\bbbone_{[c_0,c_1]}\bigl(\tfrac{\abs{x-y}}{t}\bigr)\bigl(\tfrac{x-y}{t}-\nabla\omega(D_x)+\nabla\Omega(D_y)\bigr)}\bigr].
  \end{align*}
Then
  \begin{equation}
      \int_1^\infty\norm[\big]{\Theta_{[c_0,c_1]}(t)^{\frac{1}{2}}e^{-\ri tH}\chi(H)u}^2\frac{\ud
        t}{t}\le C\norm{u}^2.
  \end{equation}
\end{theorem}

\begin{proof}
  The following construction is taken from \cite{DGe} but ultimately
  goes back to a construction of Graf, see e.g.\ \cite{Graf}. There
  exists a function $R_0\in C^\infty(\R^\nu)$ such that
  \begin{align*}
    &&\strut&&\strut&& R_0(x)&=0&\text{ for
    }\abs{x}&\le\tfrac{c_0}{2},&&\strut&&\strut&&\\
    &&\strut&&\strut&& R_0(x)&=\tfrac{1}{2}x^2+c&\text{ for
    }\abs{x}&\ge2c_1,&&\strut&&\strut&&\\
    &&\strut&&\strut&& \nabla^2R_0(x)&\ge\bbbone_{[c_0,c_1]}(\abs{x}).\hspace{-2cm}&&&&\strut&&\strut&&
  \end{align*}
  Without loss of generality, we may assume that $c_1>C_1+1$, where
  $C_1$ is the constant whose existence is ensured by
  Theorem~\ref{thm:largevel}. Choose a constant $c_2>c_1+1$ and a
  smooth function $F$ such that $F(s)=1$ for $s<c_1$ and $F(s)=0$ for
  $s\ge c_2$. Let
  \begin{equation*}
    R(x)=F(\abs{x})R_0(x).
  \end{equation*}
  Then $R$ satisfies
  \begin{align}
    \nabla^2R(x)&\ge\bbbone_{[c_0,c_1]}(\abs{x})-C\bbbone_{[C_1+1,c_2]}(\abs{x}),\label{eq:phasespaceineq}\\
    \abs{\partial^\alpha R(x)}&\le C_\alpha.\nonumber
  \end{align}
  Write $X=\tfrac{x-y}{t}-\nabla\omega(D_x)+\nabla\Omega(D_y)$ and let 
  \begin{equation*}
    \Phi(t)=\chi(H)[b(t)]\chi(H),
  \end{equation*}
  where
  \begin{equation*}
    b(t)=R(\tfrac{x-y}{t})-\tfrac{1}{2}\bigl(\inner{\nabla R(\tfrac{x-y}{t})}{X}+\hc\bigr).
  \end{equation*}
  By using Condition~\ref{cond:rho}~\eqref{cond:item:SR} and
  pseudo-differential calculus, one sees that
  \begin{equation*}
    \norm[\bigg]{\chi(H)\begin{pmatrix}0&0\\-\ri b(t)\rho(x-\cdot\,)&0
      \end{pmatrix}\chi(H)}\in O(t^{-1-\mu})
  \end{equation*}
  and hence 
  \begin{equation*}
    \chi(H)[V,\ri [b(t)]]\chi(H)\in O(t^{-1-\mu}).
  \end{equation*}
  Compute
  \begin{align*}
    \tfrac{d}{dt}b(t)\MoveEqLeft=-\tfrac{1}{t}\inner{\tfrac{x-y}{t}}{\nabla R(\tfrac{x-y}{t})}\\&+\tfrac{1}{2}\tfrac{1}{t}\bigl(\inner{\tfrac{x-y}{t}}{\nabla^2R(\tfrac{x-y}{t})X}+\hc\bigr)\\&+\tfrac{1}{t}\inner{\nabla R(\tfrac{x-y}{t})}{\tfrac{x-y}{t}}\\
    \MoveEqLeft=\tfrac{1}{2}\tfrac{1}{t}\bigl(\inner{\tfrac{x-y}{t}}{\nabla^2R(\tfrac{x-y}{t})X}+\hc\bigr),
  \end{align*}
  and by pseudo-differential calculus one sees that
  \begin{align*}
    [\omega(D_x)+\Omega(D_y),\ri b(t)]\MoveEqLeft=\tfrac{1}{2}\tfrac{1}{t}\bigl(\inner{\nabla\omega(D_x)-\nabla\Omega(D_y)}{\nabla R(\tfrac{x-y}{t})}+\hc\bigr)\\
    &-\tfrac{1}{2}\tfrac{1}{t}\bigl(\inner{\nabla\omega(D_x)-\nabla\Omega(D_y)}{\nabla^2R(\tfrac{x-y}{t})X}+\hc\bigr)\\&
    -\tfrac{1}{2}\tfrac{1}{t}\bigl(\inner{\nabla R(\tfrac{x-y}{t})}{\nabla\omega(D_x)-\nabla\Omega(D_y)}+\hc\bigr)\\&+O(t^{-2})\\
    \MoveEqLeft=-\tfrac{1}{2}\tfrac{1}{t}\bigl(\inner{\nabla\omega(D_x)-\nabla\Omega(D_y)}{\nabla^2R(\tfrac{x-y}{t})X}+\hc\bigr)\\&+O(t^{-2}),
  \end{align*}
  hence by using \eqref{eq:phasespaceineq}, it follows that
  \begin{align*}
    \MoveEqLeft\chi(H)[\bd_0b(t)]\chi(H)\\
    &=\tfrac{1}{t}\chi(H)[\inner{X}{\nabla^2R(\tfrac{x-y}{t})X}]\chi(H)+O(t^{-2})\\
    &\ge
    \tfrac{1}{t}\chi(H)\bigl[\inner[\big]{X}{\bbbone_{[c_0,c_1]}\bigl(\tfrac{\abs{x-y}}{t}\bigr)X}\bigr]\chi(H) \\
     &\quad-\tfrac{C}{t}\chi(H)\bigl[\inner[\big]{X}{\bbbone_{[C_1+1,c_2]}\bigl(\tfrac{\abs{x-y}}{t}\bigr)X}\bigr]\chi(H)+O(t^{-2})\\
  \end{align*}
  By introducing $J\in C_0^\infty(\R;[0,1])$ supported above $C_1$
  with $J\bbbone_{[C_1+1,c_2]}=\bbbone_{[C_1+1,c_2]}$ and
  $\tilde\chi\in C_0^\infty(\R)$ with $\tilde\chi\chi=\chi$ and using
  pseudo-differential calculus, the functional calculus of almost
  analytic extensions and Condition~\ref{cond:rho} \eqref{cond:item:SR} again, one gets that
  \begin{align*}
    \MoveEqLeft\tfrac{C}{t}\chi(H)\bigl[X_i\bbbone_{[C_1+1,c_2]}\bigl(\tfrac{\abs{x-y}}{t}\bigr)X_i\bigr]\chi(H)\\&\le
    \tfrac{C}{t}\chi\tilde\chi(H)\bigl[X_iJ^3\bigl(\tfrac{\abs{x-y}}{t}\bigr)X_i\bigr]\tilde\chi\chi(H)\\
&=\tfrac{C}{t}\chi(H)\bigl[J\bigl(\tfrac{\abs{x-y}}{t}\bigr)\bigr]\tilde\chi(H)\bigl[X_iJ\bigl(\tfrac{\abs{x-y}}{t}\bigr)X_i\bigr]\tilde\chi(H)\bigl[J\bigl(\tfrac{\abs{x-y}}{t}\bigr)\bigr]\chi(H)+O(t^{-2})\\
&\le\tfrac{C'}{t}\chi(H)\bigl[J^2\bigl(\tfrac{\abs{x-y}}{t}\bigr)\bigr]\chi(H)+Ct^{-2},
  \end{align*}
  where we estimated
  $\tilde\chi(H)\bigl[X_iJ\bigl(\tfrac{\abs{x-y}}{t}\bigr)X_i\bigr]\tilde\chi(H)$
  by a constant.  Putting it all together yields
  \begin{align*}
    \bD\Phi(t)\ge \tfrac{1}{t}\chi(H)\Theta_{[c_0,c_1]}(t)\chi(H)-\tfrac{C}{t}\chi(H)[J^2(\tfrac{\abs{x-y}}{t})]\chi(H)+O(t^{-1-\mu}),
  \end{align*}
  where the second term is integrable along the evolution by
  Theorem~\ref{thm:largevel}, so the result now follows from
  Lemma~\ref{lem:propest}.
\end{proof}

\begin{theorem}[Improved phase-space propagation estimate]\label{thm:improvedphase-space}
  Let $0<c_0<c_1$, $J\in C_0^\infty(c_0<\abs{x}<c_1)$, $\chi\in
  C_0^\infty(\R)$. Then for $1\le i\le\nu$
  \begin{align*}
    \label{eq:improvedphase-space}
    \int_1^\infty\norm[\big]{\bigl[\abs[\big]{J(\tfrac{x-y}{t})\bigl(\tfrac{x_i-y_i}{t}-\partial_i\omega(D_x)+\partial_i\Omega(D_y)\bigr)+\hc}\bigr]^{\frac{1}{2}}e^{-\ri
        tH}\chi(H)u}^2\frac{\ud t}{t}\le C\norm{u}^2
  \end{align*}
\end{theorem}

\begin{proof}
  For brevity, we write
  $X=\frac{x-y}{t}-\nabla\omega(D_x)+\nabla\Omega(D_y)$ and $R_0=(H_0-\lambda)\inv$ for some real $\lambda\in\rho(H_0)$.
  Let
  \begin{equation*}
    A=X^2+t^{-\delta},
  \end{equation*}
  $\delta>0$.  Note that $[J(\tfrac{x-y}{t})A^{\frac{1}{2}}]R_0$ is
  uniformly bounded in $t\ge1$. 
  
  The following identities hold as forms on $C_0^\infty(\R^\nu)$.
  \begin{equation*}
    e^{\ri t(\omega(D_x)+\Omega(D_y))}Xe^{-\ri t(\omega(D_x)+\Omega(D_y))}=\tfrac{x-y}{t},
  \end{equation*}
  \begin{equation}\label{eq:A0}
    e^{\ri t(\omega(D_x)+\Omega(D_y))}A^{\frac{1}{2}}e^{-\ri t(\omega(D_x)+\Omega(D_y))}=\bigl((\tfrac{x-y}{t})^2+t^{-\delta}\bigr)^{\frac{1}{2}}:=A_0^{\frac{1}{2}}
  \end{equation}
  and
  \begin{equation}\label{eq:J}
    e^{\ri t(\omega(D_x)+\Omega(D_y))}J(X)e^{-\ri t(\omega(D_x)+\Omega(D_y))}=J(\tfrac{x-y}{t}).
  \end{equation}
  
  That the following commutator, viewed as a form on
  $C_0^\infty(\R^\nu)$, extends by continuity to a bounded form on
  $L^2(\R^\nu)$ can be seen using pseudo-differential calculus:
  \begin{equation*}
    [X,A_0^{\frac{1}{2}}]=-[\nabla\omega(D_x),A_0^{\frac{1}{2}}]+[\nabla\Omega(D_y),A_0^{\frac{1}{2}}]=O(t^{-2+\frac{\delta}{2}}).
  \end{equation*}
  Together with the functional calculus of almost analytic extensions
  this implies that
  \begin{equation*}
    [J(X),A_0^{\frac{1}{2}}]=O(t^{-2+\frac{\delta}{2}}),
  \end{equation*}
  and hence using \eqref{eq:A0} and \eqref{eq:J} that
  \begin{equation}\label{eq:lemma6.4.ii}
    [J(\tfrac{x-y}{t}),A^{\frac{1}{2}}]=O(t^{-2+\frac{\delta}{2}}).
  \end{equation}
  Write $h=\Omega(D_y)+\omega(D_x)$. Note that
  \begin{align*}
    e^{\ri th}\bd_0A^{\frac{1}{2}}e^{-\ri th}&=e^{\ri th}[h,\ri A^{\frac{1}{2}}]e^{-\ri th}+e^{\ri th}(\tfrac{d}{dt}A^{\frac{1}{2}})e^{-\ri th}\\
    &=\tfrac{d}{dt}\bigl(e^{\ri th}A^{\frac{1}{2}}e^{-\ri th}\bigr)=\tfrac{d}{dt}A_0^{\frac{1}{2}}\\
    &=-\tfrac{1}{t}A_0^{\frac{1}{2}}-\tfrac{(2-\delta)t^{-\delta-1}}{2\bigl((\frac{x-y}{t})^2+t^{-\delta}\bigr)^{\frac12}},
  \end{align*}
  so
  \begin{equation}\label{eq:lemma6.4.iii}
    \bd_0A^{\frac{1}{2}}=-\tfrac{1}{t}A^{\frac{1}{2}}+O(t^{-1-\frac{\delta}{2}}).
  \end{equation}
  
  In addition
  \begin{equation}\label{eq:komRnulvi}
    [R_0,[X_i]]=R_0^{\frac{1}{2}+\rho_1}O(t\inv)R_0^{1-\rho_1}
  \end{equation}
  for any
  $\rho_1$, $0<\rho_1<\frac{1}{2}$ and that
  \begin{equation}\label{eq:komRnula}
    [R_0,[A^{\frac{1}{2}}]]=R_0^{\rho_2}O(t^{\frac{\delta}{2}-1})R_0^{1-\rho_2}
  \end{equation}
  for any $\rho_2$, $0<\rho_2<1$. The identity \eqref{eq:komRnula} can
  be seen e.g.\ by using \eqref{eq:komRnulvi} and the representation
  formula
  \begin{equation*}
    s^{-\frac{1}{2}}=\frac{1}{\pi}\int_0^\infty(s+y)\inv
    y^{-\frac{1}{2}}\ud y,
  \end{equation*}
  which can be verified for $s>0$ by direct computations.
  
  Let $J_1,J_2\in C_0^\infty(c_0<\abs{x}<c_1)$ such that $JJ_1=J$ and
  $J_1J_2=J_1$ and write for $i=1,\dotsc,\nu$: 
  \begin{equation*}
    B_{0,i}=R_0[J(\tfrac{x-y}{t})X_i]R_0+\hc
  \end{equation*}
  and
  \begin{equation}\label{eq:defB1}
    B_1=R_0[J_1(\tfrac{x-y}{t})A^{\frac12}J_1(\tfrac{x-y}{t})]R_0.
  \end{equation}
  
  We compute using \eqref{eq:lemma6.4.ii}, \eqref{eq:komRnulvi} and
  \eqref{eq:komRnula}:
  \begin{align*}
    \MoveEqLeft[1] B_{0,i}^2=4R_0[X_iJ(\tfrac{x-y}{t})]R_0^2[J(\tfrac{x-y}{t})X_i]R_0+O(t\inv)\\
&=4R_0^2[X_iJ^2(\tfrac{x-y}{t})X_i]R_0^2+O(t\inv)\\
    &\le CR_0^2[X_iJ_1^4(\tfrac{x-y}{t})X_i]R_0^2+Ct\inv\\
    &= CR_0^2[J_1^2(\tfrac{x-y}{t})X_i^2J_1^2(\tfrac{x-y}{t})]R_0^2+O(t\inv)\\
    &\le CR_0^2[J_1^2(\tfrac{x-y}{t})AJ_1^2(\tfrac{x-y}{t})]R_0^2+O(t^{-\delta})\\
    &=
    CR_0[J_1^2(\tfrac{x-y}{t})A^{\tfrac{1}{2}}]R_0^2[A^{\tfrac12}J_1^2(\tfrac{x-y}{t})]R_0+O(t^{-\min\{1-\frac{\delta}{2},\delta\}})\\
    &=
    CR_0[J_1(\tfrac{x-y}{t})A^{\tfrac{1}{2}}J_1(\tfrac{x-y}{t})]R_0^2[J_1(\tfrac{x-y}{t})A^{\tfrac12}J_1(\tfrac{x-y}{t})]R_0+O(t^{-\min\{1-\frac{\delta}{2},\delta\}})\\
    &=CB_1^2+O(t^{-\kappa}),
  \end{align*}
  where $\kappa=\min\{1-\frac{\delta}{2},\delta\}$. By the matrix
  monotonicity of $\lambda\mapsto \lambda^{\frac{1}{2}}$
  \cite[Sec. 2.2.2]{BR}, we deduce that
  \begin{equation}\label{eq:lemma6.4.iv}
    \abs{B_{0,i}}\le CB_1+Ct^{-\frac{\kappa}{2}}.
  \end{equation}
  
  Now let \begin{equation}\label{eq:improvedobservable}
    \Phi(t)=-\chi(H)[J(\tfrac{x-y}{t})A^{\frac{1}{2}}J(\tfrac{x-y}{t})]\chi(H)
  \end{equation}
  It follows from \eqref{eq:lemma6.4.ii} that
  \begin{equation}\label{eq:lemma6.4.i}
    \Phi(t)=-\chi(H)[J(\tfrac{x-y}{t})^2A^{\frac{1}{2}}]\chi(H)+O(t^{-2+\frac{\delta}{2}})
  \end{equation}
  is uniformly bounded for $t>1$.
  
  We compute
  \begin{align}\label{eq:improvedobservableheisenberg}
    \MoveEqLeft[1]-\bD\Phi(t)=\\&\chi(H)[V,\ri [J(\tfrac{x-y}{t})A^{\frac{1}{2}}J(\tfrac{x-y}{t})]]\chi(H)+
    \chi(H)\bigl[\bd_0\bigl(J(\tfrac{x-y}{t})A^{\frac{1}{2}}J(\tfrac{x-y}{t})\bigr)\bigr]\chi(H)\nonumber
  \end{align}
  Using Condition~\ref{cond:rho}~\eqref{cond:item:SR} we see that
  \begin{equation*}
    \chi(H)[V,\ri [J(\tfrac{x-y}{t})A^{\frac{1}{2}}J(\tfrac{x-y}{t})]]\chi(H)=O(t^{-1-\mu}). 
  \end{equation*}
Indeed,
  \begin{align*}
    \MoveEqLeft
    \chi(H)[V,\ri [J(\tfrac{x-y}{t})A^{\frac{1}{2}}J(\tfrac{x-y}{t})]]\chi(H)\\
    &=\chi(H)
    \Bigl(\begin{smallmatrix}0&0&\\-\ri J(\tfrac{x-y}{t})A^{\frac{1}{2}}J(\tfrac{x-y}{t})v&0
    \end{smallmatrix}\Bigr)
    \chi(H)+\hc\\
    &=\chi(H)(H_0-\lambda)R_0\Bigl(\begin{smallmatrix}0&0&\\-\ri \bigl(A^{\frac{1}{2}}J(\tfrac{x-y}{t})+O(t^{-2+\frac{\delta}{2}})\bigr)J(\tfrac{x-y}{t})v&0
    \end{smallmatrix}\Bigr)
    \chi(H)+\hc
  \end{align*}
  Now by Condition~\ref{cond:rho}~\eqref{cond:item:SR} we have that
  $\norm{J(\tfrac{x-y}{t})v}=O(t^{-1-\mu})$ and hence we also have
  that
  \begin{equation*}
    R_0\Bigl(\begin{smallmatrix}
      0&0&\\-\ri
      \bigl(A^{\frac{1}{2}}J(\tfrac{x-y}{t})+O(t^{-2+\frac{\delta}{2}})\bigr)J(\tfrac{x-y}{t})v&0
    \end{smallmatrix}\Bigr)=O(t^{-1-\mu}).
  \end{equation*}

  Note that 
  \begin{equation}
    \label{eq:mixedtermsheisenberg}
    \bd_0 J(\tfrac{x-y}{t})=-\tfrac{1}{t}\nabla J(\tfrac{x-y}{t})\cdot X+O(t^{-2})
  \end{equation}
  and using \eqref{eq:lemma6.4.iii} and \eqref{eq:lemma6.4.iv}
  (cf. \eqref{eq:defB1}),
  \begin{align*}
    \MoveEqLeft-\chi(H)[J(\tfrac{x-y}{t})(\bd_0 A^{\frac{1}{2}})J(\tfrac{x-y}{t})]\chi(H)\\&\ge\tfrac{C_0}{t}\chi(H)[\abs{J(\tfrac{x-y}{t})X_i+\hc}]\chi(H)-Ct^{-1-\frac{\kappa}{2}},
  \end{align*}
  where $\kappa$ is from \eqref{eq:lemma6.4.iv}. Again we compute using \eqref{eq:lemma6.4.ii}:
  \begin{align*}
    \MoveEqLeft[1]
    R_0[\nabla J(\tfrac{x-y}{t})\cdot X A^{\frac{1}{2}}J(\tfrac{x-y}{t})]R_0+\hc\\
    &=R_0[J_2(\tfrac{x-y}{t})X\cdot \nabla J(\tfrac{x-y}{t})J(\tfrac{x-y}{t})A^{\frac{1}{2}}J_2(\tfrac{x-y}{t})]R_0+\hc+O(t\inv)\\
    &=\sum_{i=1}^\nu R_0[J_2(\tfrac{x-y}{t})A^{\frac{1}{2}}X_iA^{-\frac{1}{2}}\partial_i J(\tfrac{x-y}{t})J(\tfrac{x-y}{t})A^{\frac{1}{2}}J_2(\tfrac{x-y}{t})]R_0+\hc+O(t\inv)\\
    &\le CR_0[J_2(\tfrac{x-y}{t})AJ_2(\tfrac{x-y}{t})]R_0+Ct\inv\\
    &\le
    CR_0[J_2(\tfrac{x-y}{t})X^2J_2(\tfrac{x-y}{t})]R_0+O(t^{-\min\{1,\delta\}})\\
    &\le
    CR_0[\inner{X}{J_2^2(\tfrac{x-y}{t})X}]R_0+Ct^{-1}.
  \end{align*}
  Hence (cf. \eqref{eq:mixedtermsheisenberg})
  \begin{align}
    \MoveEqLeft-\chi(H)\bigl[\bd_0\bigl(J(\tfrac{x-y}{t})A^{\frac{1}{2}}J(\tfrac{x-y}{t})\bigr)\bigr]\chi(H)\nonumber\\
    &=\chi(H)[(\bd_0J(\tfrac{x-y}{t}))A^{\frac{1}{2}}J(\tfrac{x-y}{t})]\chi(H)+\hc\nonumber\\
    &\quad+\chi(H)[J(\tfrac{x-y}{t})(\bd_0A^{\frac{1}{2}})J(\tfrac{x-y}{t})]\chi(H)\nonumber\\
    \begin{split}
    &\ge\tfrac{C_0}{t}\chi(H)[\abs{J(\tfrac{x-y}{t})X_i+\hc}]\chi(H)\\
    &\quad-\tfrac{C}{t}\chi(H)[\inner{X}{J_2^2(\tfrac{x-y}{t})X}]\chi(H)+O(t^{-1-\gamma})\label{eq:oldpropestimate}
  \end{split}
  \end{align}
  for some $\gamma>0$. Since by Theorem~\ref{thm:phase-space} the
  second term in the r.h.s. of \eqref{eq:oldpropestimate} is integrable
  along the evolution, the theorem follows from
  Lemma~\ref{lem:propest}.
\end{proof}

\begin{theorem}[Minimal velocity estimate]\label{thm:minvelest}
  Assume that $(P_0,\lambda_0)\in \R^{\nu+1}$ satisfies that
  $\lambda_0\in\R\setminus(\theta(P_0)\cup\sigma_{\mathrm{pp}}(P_0))$. Then
  there exists an $\epsilon>0$, a neighbourhood $N$ of
  $(P_0,\lambda_0)$ and a function $\chi\in C_0^\infty(\R^{\nu+1})$
  such that $\chi=1$ on $N$ and
  \begin{equation*}
    \int_1^\infty\norm[\Big]{\begin{pmatrix} 1 & 0 \\ 0 & \bbbone_{[0,\epsilon]}(\tfrac{\abs{x-y}}{t})\end{pmatrix}\,e^{-\ri
        tH}\chi(\uP,H) u}^2\frac{\ud t}{t}\le C\norm{u}^2
  \end{equation*}
\end{theorem}

\begin{proof}
  By Theorem~\ref{thm:lowsemcontMourre}, it follows that there exists
  a neighbourhood $\cO$ of $P_0$ and a function $f$ with $f=1$ in a
  neighbourhood of $\lambda_0$ such that
  \begin{equation}\label{eq:P-cont-mourre}
    f(H(P))[H(P),\ri A_{P_0}]^\circ f(H(P))\ge Cf^2(H(P))
  \end{equation}
  for all $P$ in $\cO$. Let $\chi\in
  C_0^\infty(\R^{\nu+1};[0,1])$ be supported in $\cO\times
  \{\lambda\,|\,f(\lambda)=1\}$ and $\chi=1$ in a neighbourhood $N$ of
  $(P_0,\lambda_0)$. It follows that
  \begin{equation}\label{eq:globalmourre}
    \chi(P,H(P))[H(P),\ri A_{P_0}]^\circ\chi(P,H(P))\ge\tfrac{C}{2}\chi^2(P,H(P)).
  \end{equation}
  Let $q\in C_0^\infty(\{\abs{x}\le2\epsilon\})$ satisfy $0\le q\le1$,
  $q=1$ in a neighbourhood of $\{\abs{x}\le\epsilon\}$ for some
  $\epsilon>0$ to be specified later on. Write
  \begin{equation*}
    Q(t)=\begin{pmatrix}1&0\\0&q(\tfrac{x}{t})
    \end{pmatrix}.
  \end{equation*}
  Let
  \begin{equation*}
    \Phi(t)=\!\int^\oplus\!\!\chi(P,H(P))Q(t)\frac{A_{P_0}}{t}Q(t)\chi(P,H(P))\ud P.
  \end{equation*}
  Taking into account the support of $q$ and that $v_{P_0}$ is
  $\omega$-bounded, and using pseudo-differential calculus, it is easy
  to see that $\Phi(t)$ is uniformly bounded.

  We compute the Heisenberg derivative:
  \begin{align*}
    \bD\Phi(t)\MoveEqLeft[1]=\!\int^\oplus\!\!\chi(P,H(P))[\bd_{0}^Pq(\tfrac{x}{t})]\frac{A_{P_0}}{t}Q(t)\chi(P,H(P))\ud
    P+\hc\\
    &+\!\int^\oplus\!\!\chi(P,H(P))[V,\ri Q(t)]\frac{A_{P_0}}{t}Q(t)\chi(P,H(P))\ud
    P+\hc\\
    &+\frac{1}{t}\!\int^\oplus\!\!\chi(P,H(P))Q(t)[H(P),\ri A_{P_0}]Q(t)\chi(P,H(P))\ud
    P\\
    &-\frac{1}{t}\!\int^\oplus\!\!\chi(P,H(P))Q(t)\frac{A_{P_0}}{t}Q(t)\chi(P,H(P))\ud
    P\\
    \MoveEqLeft[1]=R_1+R_2+R_3+R_4,
  \end{align*}
  where $\bd_0^P=[\Omega(P-D_x)+\omega(D_x),\,\cdot\,]+\frac{\ud}{\ud t}$.

  By the same arguments as before it follows that
  $\frac{A_{P_0}}{t}Q(t)\chi(P,H(P))$ is uniformly bounded. Using
  pseudo-differential calculus gives
  \begin{align*}
    \MoveEqLeft[1] R_1=\\
    &-\frac{1}{t}\!\int^\oplus\!\!\chi(P,H(P))\bigl[\inner{\tfrac{x}{t}-\nabla\omega(D_x)+\nabla\Omega(P-D_x)}{\nabla
      q(\tfrac{x}{t})}\bigr]\frac{A_{P_0}}{t}Q(t)\chi(P,H(P))\ud
    P\\
    &\quad+\hc+O(t^{-2}).
  \end{align*}
  Let
  \begin{equation*}
    B_1=-\int^\oplus\!\!\chi(P,H(P))[\inner{\tfrac{x}{t}-\nabla\omega(D_x)+\nabla\Omega(P-D_x)}{\nabla
      q(\tfrac{x}{t})}]\ud
    P
  \end{equation*}
  and
  \begin{equation*}
    B_2=\!\int^\oplus\!\!\chi(P,H(P))Q(t)\frac{A_{P_0}}{t}\ud
    P.
  \end{equation*}
  Then 
  \begin{equation*}
    R_1=\tfrac{1}{t}B_1B_2^*+\tfrac{1}{t}B_2B_1^*\ge-\epsilon_0\inv\tfrac{1}{t}B_1B_1^*-\epsilon_0\tfrac{1}{t}B_2B_2^*.
  \end{equation*}
  Now by Theorem~\ref{thm:phase-space}, we get that
  $\tfrac{1}{t}B_1B_1^*$ is integrable along the evolution.  Using
  pseudo-differential calculus and functional calculus of almost
  analytic extensions one can verify that
  \begin{equation}\label{eq:commutatorchiQ}
    [\chi(P,H(P)),Q(t)]=(H_0(P)-R)^{-1+\rho}O(t^{-1})(H_0(P)-R)^{-\tfrac{1}{2}-\rho}
  \end{equation}
  for any $R\in\R\setminus\sigma(H_0(P))$ and any $\rho$,
  $0\le\rho\le\tfrac{1}{2}$. Hence it follows by introducing cutoff
  functions $\tilde\chi\in C_0^\infty(\R^{\nu+1})$ and $\tilde q\in
  C_0^\infty(\R^\nu)$ with $\tilde\chi\chi=\chi$ and $\tilde qq=q$
  that
  \begin{align}
    \MoveEqLeft-\tfrac{1}{t}B_2B_2^*=-\frac{1}{t}\!\int^\oplus\!\!Q(t)\chi\tilde\chi(P,H(P))[\tilde
    q(\tfrac{x}{t})]\frac{A_{P_0}^2}{t^2}[\tilde
    q(\tfrac{x}{t})]\tilde\chi\chi(P,H(P))Q(t)\ud
    P\nonumber\\&+O(t^{-2})\nonumber\\
    \MoveEqLeft[1]\ge-\frac{C_1}{t}\!\int^\oplus\!\!Q(t)\chi^2(P,H(P))Q(t)\ud
    P+O(t^{-2})\nonumber\\
    \MoveEqLeft[1]=-\frac{C_1}{t}\!\int^\oplus\!\!\chi(P,H(P))Q^2(t)\chi(P,H(P))\ud
    P+O(t^{-2})\label{eq:B_2}
  \end{align}
  By Condition~\ref{cond:rho}~\eqref{cond:item:SR} it follows that
  $\Bigl(
  \begin{smallmatrix}0&0\\
    \ri (1-q(\tfrac{x}{t}))\left|\rho\right\rangle&0
  \end{smallmatrix}\Bigr)\in O(t^{-1-\mu})$
  and hence 
  \begin{equation}
    \label{eq:R_2}
    R_2\in O(t^{-1-\mu})
  \end{equation}
  Using \eqref{eq:globalmourre} and \eqref{eq:commutatorchiQ} twice, we
  see that
  \begin{align}
    R_3&=\frac{1}{t}\!\int^\oplus\!\!Q(t)\chi(P,H(P))[H(P),\ri A_{P_0}]\chi(P,H(P))Q(t)\ud
    P+O(t^{-2})\nonumber\\
    &\ge
    \frac{C_2}{t}\!\int^\oplus\!\!Q(t)\chi^2(P,H(P))Q(t)\ud P+O(t^{-2})\nonumber\\
    &\ge
    \frac{C_2}{t}\!\int^\oplus\!\!\chi(P,H(P))Q(t)^2\chi(P,H(P))\ud
    P+O(t^{-2}).\label{eq:R_3}
  \end{align}
  Again using the cutoff functions and pseudo-differential calculus
  and taking into account the support of $q$, we see that
  \begin{align*}
    \MoveEqLeft\pm\chi(P,H(P))Q(t)\frac{A_{P_0}}{t}Q(t)\chi(P,H(P))\\
    &=\pm Q(t)\chi\tilde\chi(P,H(P))[\tilde
    q(\tfrac{x}{t})]\frac{A_{P_0}}{t}[\tilde
    q(\tfrac{x}{t})]\tilde\chi\chi(P,H(P))Q(t)\pm O(t\inv)\\
    &\le \epsilon C_3Q(t)\chi^2(P,H(P))Q(t)+O(t\inv)\\
    &=\epsilon C_3\chi(P,H(P))Q(t)^2\chi(P,H(P))+O(t\inv)
  \end{align*}
  so 
  \begin{equation}
    \label{eq:R_4}
    R_4\ge-\frac{C_3\epsilon}{t}\!\int^\oplus\!\!\chi(P,H(P))Q(t)^2\chi(P,H(P))\ud P+O(t^{-2}).
  \end{equation}
  Putting \eqref{eq:B_2}, \eqref{eq:R_2}, \eqref{eq:R_3} and
  \eqref{eq:R_4} together, we see that
  \begin{align*}
    \bD\Phi(t)\MoveEqLeft\ge\frac{-\epsilon_0C_1+C_2-\epsilon
      C_3}{t}\!\int^\oplus\!\!\chi(P,H(P))Q(t)^2\chi(P,H(P))\ud
    P\\
    &-\frac{1}{\epsilon t}B_1B_1^*+O(t^{-1-\mu}).
  \end{align*}
  Now choosing $\epsilon$ and $\epsilon_0$ so small that
  $-\epsilon_0C_1+C_2-\epsilon C_3>0$ together with
  Lemma~\ref{lem:propest} yields the result.
\end{proof}

\section[Asymptotic observable and asymptotic completeness]{The
  asymptotic observable and asymptotic
  completeness}\label{sec:asympobs}
Recall the notation $[\,\cdot\,]$ from \eqref{eq:bracketnotation}.
\begin{theorem}[Asymptotic observable]\label{thm:asympobs}
  Let $p\in C^\infty(\R^\nu)$ satisfy that $p(x)\le p(y)$ for
  $\abs{x}\le\abs{y}$, $p(x)=0$ for $\abs{x}\le\tfrac{1}{2}$ and
  $p(x)=1$ for $\abs{x}\ge1$. Define
  $p_\delta(x)=p(\tfrac{x}{\delta})$. Then the limits
  \begin{align}
    P^+_\delta(H)&=\slim_{t\to\infty} e^{\ri tH}[p_\delta(\tfrac{x-y}{t})]e^{-\ri tH},\label{eq:Pplusdelta}\\
    P^+_0(H)&=\slim_{\delta\to0}P^+_\delta(H),\label{eq:Pplus}\\
    P^+_\delta(H_0,H)&=\slim_{t\to\infty}
    e^{\ri tH}[p_\delta(\tfrac{x-y}{t})]e^{-\ri tH_0},\nonumber\\
    P^+_\delta(H,H_0)&=\slim_{t\to\infty}
    e^{\ri tH_0}[p_\delta(\tfrac{x-y}{t})]e^{-\ri tH}\nonumber
  \end{align}
  exist and $P^+_0(H)$ is a projection. 
\end{theorem} 

\begin{remark}\label{rem:asympobs}
  Note that $\delta\mapsto P^+_\delta(H)$ is increasing in the sense that
  $P^+_\delta(H)\le P^+_{\delta'}(H)$ for $0<\delta'<\delta$. We leave
  it to the reader to verify that the definition of $P^+_0(H)$ is
  independent of the choice of $p$, and that one in fact could have
  chosen any family of functions $\{p_\delta\}$ satisfying
  $p_\delta(x)\le p_\delta(y)$ for $\abs{x}\le\abs{y}$,
  $p_\delta(x)=0$ for $\abs{x}\le\tfrac{\delta}{2}$ and
  $p_\delta(x)=1$ for $\abs{x}\ge\delta$.
\end{remark}

\begin{proof} We will prove the statements about $P^+_\delta(H)$ and
  $P^+_0(H)$. The statements about $P^+_\delta(H_0,H)$ and $P^+_\delta(H,H_0)$ are
  proved completely analogously to that of $P^+_\delta(H)$.

  Let
  \begin{equation*}
    \Phi(t)=-\chi(H)[p_\delta(\tfrac{x-y}{t})]\chi(H),
  \end{equation*}
  and calculate using pseudo-differential calculus
  \begin{equation*}
    \bd_0p_\delta(\tfrac{x-y}{t})=-\tfrac{1}{2}\tfrac{1}{t}\Bigl(\bigl(\tfrac{x-y}{t}-\nabla\omega(D_x)+\nabla\Omega(D_y)\bigr)\cdot\nabla
    p_\delta(\tfrac{x-y}{t})+\hc\Bigr)+O(t^{-2}).
  \end{equation*}
  This in combination with
  Condition~\ref{cond:rho}~\eqref{cond:item:SR} gives
  \begin{equation*}
    \bD\Phi(t)=\tfrac{1}{t}\chi(H)[\tfrac{1}{2}X\cdot\nabla
    p_\delta(\tfrac{x-y}{t})+\hc]\chi(H)+O(t^{-\min\{1+\mu,2\}}),
  \end{equation*}
  where $X=\tfrac{x-y}{t}-\nabla\omega(D_x)+\nabla\Omega(D_y)$, so
  Theorem~\ref{thm:improvedphase-space} in combination with
  Lemma~\ref{lem:asympobs} gives the existence of the limit
  \eqref{eq:Pplusdelta}.

  The existence of the weak limit $\w
  P^+_0(H)=\wlim_{\delta\to0}P^+_\delta(H)$ is obvious. Moreover, for every
  $\delta>0$, it is clear from Lemma~\ref{lem:proj} that the strong
  limit $\displaystyle\slim_{n\to\infty}P^+_{\frac{\delta}{2^n}}(H)$ exists, is a
  projection and equals $\w P^+_0(H)$. The inequality $P^+_\delta(H)^2\le
  P^+_\delta(H)$ implies
  \begin{align*}
    \lim_{\delta\to0}\norm[\big]{(\w
      P^+_0(H)-P^+_\delta(H))u}^2&=\lim_{\delta\to0}\inner[\big]{(\w P^+_0(H)+P^+_\delta(H)^2-2P^+_\delta(H))u}{u}\\
    &\le\lim_{\delta\to0}\inner[\big]{(\w P^+_0(H)-P^+_\delta(H))u}{u}=0.
  \end{align*}
  This finishes the argument.
\end{proof}

\begin{proposition}\label{prop:SigmaTheta}
  Let $\Sigma=
  \{(P,\lambda)\in\R^{\nu+1}\,|\,\lambda\in\sigma_\mathrm{pp}(H(P))\}$
  denote the set in energy-momentum space consisting of eigenvalues
  for the fibered Hamiltonian and
  $\Theta=\{(P,\lambda)\in\R^{\nu+1}\,|\,\lambda\in\theta(P)\}$
  the corresponding set of thresholds. Then $\Sigma\cup\Theta$ is a
  closed set of Lebesgue measure $0$. Moreover,
  $(\Sigma\cup\Theta)(P)=\sigma_{\mathrm{pp}}(P)\cup\theta(P)$ is at
  most countable.
\end{proposition}

\begin{proof}By the usual arguments, Theorems~\ref{thm:C2} and
  \ref{thm:Mourre} imply that eigenvalues of $H(P)$ can only
  accumulate at thresholds (see e.g. \cite{ABG} for details), and by
  analyticity, the threshold set $\theta(P)$ is at most
  countable. Hence, if $\Sigma\cup\Theta$ is closed, it is in particular
  of measure $0$. 

  Let $(P_0,\lambda_0)\not\in\Sigma\cup\Theta$. Then by
  Theorem~\ref{thm:lowsemcontMourre}, there are neighbourhoods $\cO$ of
  $P_0$ and $I$ of $\lambda_0$ such that for all $P\in \cO$, a strict
  Mourre estimate holds for $H(P)$ on the energy interval $I$ with
  conjugate operator $A_{P_0}$ given as in Theorem~\ref{thm:Mourre}
  and $H(P)$ is of class $C^2(A_{P_0})$ by Theorem~\ref{thm:C2}, which
  by the Virial Theorem implies that there are no eigenvalues for
  $H(P)$ in $I$ for any $P\in \cO$. Clearly,
  \begin{equation*}
    \Theta=\{(P,\lambda)\in\R^{\nu+1}\,|\,\exists k\in\R^\nu\colon\lambda=\Omega(P-k)+\omega(k),\nabla\omega(k)-\nabla\Omega(P-k)=0\}
  \end{equation*}
  is a closed set. Hence, possibly after chosing smaller $\cO$ and $I$,
  $\cO\times I$ is a neighbourhood of $(P_0,\lambda_0)$ which does not
  intersect $\Sigma\cup\Theta$.
\end{proof}

Let $\cH_\bnd=E_{\Sigma\cup\Theta}((\uP,H))\cH$ and similarly
$\cH_{0,\bnd}=E_{\Sigma_0\cup\Theta}((\uP,H_0))\cH$, where we by
$E_\cB(\uP,H)$ resp. $E_\cB(\uP,H_0)$ denote the spectral projection
for the pair of commuting, self-adjoint operators of some Borel set
$\cB\in\R^{\nu+1}$. We remark that if we for a fixed $P$ take the fiber
$(\Sigma\cup\Theta)(P)=\{\lambda\,|\,(\lambda,P)\in\Sigma\cup\Theta\}$,
then we have $E_{(\Sigma\cup\Theta)(P)}(H(P))=\bbbone_\pp(H(P))$.

\begin{theorem}\label{thm:Hbnd}With $\cH_\bnd$ and $P^+_0(H)$ given as
  above, we have $\cH_\bnd=(1-P^+_0(H))\cH$.
\end{theorem}

\begin{proof}
  Let $(\lambda_0, P_0)\in\R^{\nu+1}\setminus (\Sigma\cup\Theta)$. Let the
  neighbourhood $N$ and $\epsilon>0$ be those of Theorem~\ref{thm:minvelest} corresponding to the point $(\lambda_0,P_0)$. Let $\psi\in
  E_N(\uP,H)\cH$. Then by Theorem~\ref{thm:minvelest}, there
  exists a sequence $t_n\to\infty$ such that
  \begin{align*}
    \psi&=e^{\ri t_n H}[p_\epsilon(\tfrac{x-y}{t_n})]e^{-\ri t_n
      H}\psi+e^{\ri t_n H}\begin{pmatrix} 1&0 \\ 0 &
      1-p_\epsilon(\tfrac{x-y}{t_n}) \end{pmatrix} e^{-\ri t_n
      H} \psi \to P^+_\epsilon(H)\psi+0,
  \end{align*}
  which implies that $\psi\in P^+_0(H)\cH$. As the span of such $\psi$ is
  dense in $\cH_\bnd^\perp$ and $P^+_0(H)\cH$ is closed, this implies that
  $\cH_\bnd\supset(1-P^+_0(H))\cH$.

  By Proposition~\ref{prop:SigmaTheta}, $\Sigma\cup\Theta$ may be
  written as an at most countable union of graphs $\Sigma_i$ of Borel
  functions from (subsets of) $\R^\nu$ to $\R$ (see
  \cite[Th\'e{}or\`e{}me 21, p. 226]{AnalSets}). Let
  $\phi=U\int^\oplus\phi_P\ud P\in\cH$. Then
  $\psi=E_{\Sigma_j}(\uP,H)\phi=U\int^\oplus
  E_{\Sigma_j(P)}(H)\phi_P\ud P$. This implies that $\psi$ can be
  written as
  \begin{equation*}
    \psi=U\int^\oplus\!\!\psi_P\ud P,
  \end{equation*}
  where $\psi_P$ is an eigenvector for $H(P)$ with eigenvalue
  $\Sigma_j(P)$. Note that this ensures that $\psi_P$ is Borel as a
  function of $P$. Now
  \begin{align*}
    P^+_\delta(H)\psi&=\slim_{t\to\infty}e^{\ri tH}[p_\delta(\tfrac{x-y}{t})]e^{-\ri tH}\psi\\
    &=\slim_{t\to\infty}U\int^\oplus
    e^{\ri tH(P)}[p_\delta(\tfrac{x}{t})]e^{-\ri tH(P)}\psi_P\ud P\\
    &=\slim_{t\to\infty}e^{\ri tH}U\int^\oplus
    [p_\delta(\tfrac{x}{t})]e^{-\ri t\Sigma_j(P)}\psi_P\ud P,
  \end{align*}
  where the last integrand goes pointwise to $0$ and hence by the
  dominated convergence theorem, the limit is $0$. As $\delta$ was
  arbitrary, this shows that $P^+_0(H)\psi=0$.

  Since the span of the set of $\psi$ we have covered is dense in
  $\cH_\bnd$ and $P^+_0(H)$ is closed, we conclude that
  $\cH_\bnd\subset(1-P^+_0(H))\cH$.
\end{proof}

\begin{theorem}[Existence of wave operators]\label{thm:waveop}
  The wave operator $W^+\colon
  \cH\mapsto \cH$ given by
  \begin{equation*}
    W^+=\slim_{t\to\infty}e^{\ri tH}e^{-\ri tH_0}P^+_0(H_0),
  \end{equation*}
  exists, where $P^+_0(H_0)$ is the projection onto $\{0\}\oplus
  L^2(\R^{2\nu})=\cH_{0,\bnd}^\perp$.
\end{theorem}

\begin{proof}
  From Theorem~\ref{thm:asympobs} and Theorem~\ref{thm:Hbnd} with $H=H_0$ it follows that $P^+_0(H_0)$
  can be given as in Theorem~\ref{thm:asympobs}, and by passing to the fibered
  representation, it is easy to see that the assumptions on $\Omega$
  and $\omega$ imply that $\cH_{0,\bnd}=L^2(\R^{\nu})\oplus\{0\}$.

  By Theorem~\ref{thm:asympobs},
  \begin{equation*}
    e^{\ri tH}[p_\delta(\tfrac{x-y}{t})]e^{-\ri tH_0}=e^{\ri tH}e^{-\ri tH_0}e^{\ri tH_0}[p_\delta(\tfrac{x-y}{t})]e^{-\ri tH_0}
  \end{equation*}
  tends strongly to $P^+_\delta(H_0,H)$ when $t\to\infty$.  On the other hand,
  \begin{equation*}
    e^{\ri tH_0}[p_\delta(\tfrac{x-y}{t})]e^{-\ri tH_0}
  \end{equation*}
  tends strongly to $P^+_\delta(H_0)$ in the same limit.  This implies
  that
  \begin{equation*}
    P^+_\delta(H_0,H)=\slim_{t\to\infty}(e^{\ri tH}e^{-\ri tH_0})P^+_\delta(H_0)
  \end{equation*}
  exists. As $\delta>0$ was arbitrary, the limit
  $\slim_{t\to\infty}(e^{\ri tH}e^{-\ri tH_0})$ exists on
  \begin{equation*}
  \bigcup_{\delta>0}\Ran P^+_\delta(H_0)
\end{equation*}
  and hence on
  $\overline{\bigcup_{\delta>0}\Ran P^+_\delta(H_0)}=\Ran P^+_0(H_0)$.
\end{proof}

\begin{remark}
  By the proof of Theorem~\ref{thm:waveop},
  \begin{equation*}
  P^+_0(H_0,H)=\slim_{\delta\to0}P^+_\delta(H_0,H)
\end{equation*}
exists. By a
  completely analogous argument, one may prove that also
  \begin{equation*}
  P^+_0(H,H_0)=\slim_{\delta\to0}P^+_\delta(H,H_0)
\end{equation*}
exists.
\end{remark}

\begin{theorem}[Geometric asymptotic completeness]\label{thm:geom}
  With $W^+$ as in Theorem~\ref{thm:waveop}, $\Ran W^+=P^+_0(H)\cH$.
\end{theorem}

\begin{proof}
  Consider
  \begin{align}
    \MoveEqLeft
    e^{itH}e^{-itH_0}e^{itH_0}[P_\delta(\tfrac{x-y}{t})]e^{-itH_0}e^{itH_0}[P_\delta(\tfrac{x-y}{t})]e^{-itH_0}=\label{eq:geo3}\\
    &e^{itH}[P_\delta(\tfrac{x-y}{t})]e^{-itH}e^{itH}e^{-itH_0}e^{itH_0}[P_\delta(\tfrac{x-y}{t})]e^{-itH_0},\label{eq:geo4}
  \end{align}
  and observe that \eqref{eq:geo3} tends to $W^+$ and \eqref{eq:geo4}
  tends to $P_0^+(H)W^+$ in the limit $t\to\infty$, $\delta\to0$,
  which proves that $\Ran W^+\subset P^+_0(H)\cH$. For the other
  inclusion, we similarly compute
  \begin{align}
    \MoveEqLeft e^{itH}[P_\delta(\tfrac{x-y}{t})]e^{-itH}e^{itH}[P_\delta(\tfrac{x-y}{t})]e^{-itH}=\label{eq:geo1}\\
&e^{itH}e^{-itH_0}e^{itH_0}[P_\delta(\tfrac{x-y}{t})]e^{-itH_0}e^{itH_0}[P_\delta(\tfrac{x-y}{t})]e^{-itH}\label{eq:geo2}
  \end{align}
  and observe that \eqref{eq:geo1} tends to $P_0^+(H)$ while
  \eqref{eq:geo2} tends to $W^+P_0^+(H,H_0)$ in the same limit,
  which proves $\Ran P^+_0(H)\subset\Ran W^+$.
\end{proof}

Theorem~\ref{thm:main} now follows from
Proposition~\ref{prop:SigmaTheta}, Theorem~\ref{thm:Hbnd} and
Theorem~\ref{thm:geom}.

\section*{Acknowledgements}

The authors would like to thank J. Hoffmann-J\o{}rgensen for pointing
out the existence of \cite[Th\'e{}or\`e{}me 21,
p. 226]{AnalSets}). Morten Grud Rasmussen acknowledges the hospitality
of D\'epartement de Math\'ematiques, Universit\'e Paris-Sud and the
hospitality and support of The Erwin Schr\"odinger International
Institute for Mathematical Physics, Vienna, where part of this work
was done.

\appendix

\section{Lemmata related to propagation estimates}
\label{app:A}
\setcounter{section}{1} \renewcommand{\thesection}{\Alph{section}}
\setcounter{theorem}{0} 
For easy reference, we list the following
lemmata, which are taken from the appendix of \cite{DGe}. The first lemma
which is used to prove the propagation estimates, is a version of the
Putnam-Kato theorem developed by Sigal--Soffer \cite{SiSo}.

\begin{lemma}\label{lem:propest}
  Let $H$ be a self-adjoint operator and $\bD$ the corresponding
  Heisenberg derivative
  \begin{equation*}
    \bD=\frac{\ud}{\ud t}+[H,\ri \,\cdot\,].
  \end{equation*}
  Suppose that $\Phi(t)$ is a uniformly bounded family of self-adjoint
  operators. Suppose that there exist $C_0>0$ and operator valued
  functions $B(t)$ and $B_i(t)$, $i=1,\dotsc,n$, such that 
  \begin{equation*}
    \bD\Phi(t)\ge C_0B^*(t)B(t)-\sum_{i=1}^nB_i^*(t)B_i(t),
  \end{equation*}
  \begin{equation*}
    \int_1^\infty\norm{B_i(t)e^{-\ri tH}\phi}^2\ud t\le
    C\norm{\phi}^2,\quad i=1,\dotsc,n.
  \end{equation*}
  Then there exists $C_1$ such that 
  \begin{equation*}
    \int_1^\infty\norm{B(t)e^{-\ri tH}\phi}^2\ud t\le C_1\norm{\phi}^2.
  \end{equation*}
\end{lemma}

The next lemma shows how to use propagation estimates to prove the
existence of asymptotic observables and is a version of Cook's method
due to Kato.

\begin{lemma}\label{lem:asympobs}
  Let $H_1$ and $H_2$ be two self-adjoint operators. Let $_2\bD_1$ be
  the corresponding asymmetric Heisenberg derivative:
  \begin{equation*}
    _2\bD_1\Phi(t)=\frac{\ud}{\ud t}\Phi(t)+\ri H_2\Phi(t)-\ri \Phi(t)H_1.
  \end{equation*}
Suppose that $\Phi(t)$ is a uniformly bounded function with values in
self-adjoint operators. Let $\cD_1\subset\cH$ be a dense
subspace. Assume that 
\begin{align*}
  \abs{\inner{\psi_2}{_2\bD_1\Phi(t)\psi_1}}&\le\sum_{i=1}^n\norm{B_{2i}(t)\psi_2}\norm{B_{1i}(t)\psi_1},\\
  \int_1^\infty\norm{B_{2i}(t)e^{-\ri tH_2}\phi}^2\ud
  t&\le\norm{\phi}^2,\quad\phi\in\cH,\, i=1,\dotsc,n,\\
  \int_1^\infty\norm{B_{1i}(t)e^{-\ri tH_1}\phi}^2\ud t&\le
  C\norm{\phi}^2,\quad\phi\in\cD_1,\, i=1,\dotsc,n.
\end{align*}
Then the limit
\begin{equation*}
  \slim_{t\to\infty}e^{\ri tH_2}\Phi(t)e^{-\ri tH_1}
\end{equation*}
exists.
\end{lemma}
The final lemma gives us the actual asymptotic observable.
\begin{lemma}\label{lem:proj}
  Let $Q_n$ be a commuting sequence of self-adjoint operators such
  that:
  \begin{equation*}
    0\le Q_n\le1,\quad Q_n\le Q_{n+1},\quad Q_{n+1}Q_n=Q_n.
  \end{equation*}
  Then the limit
  \begin{equation*}
    Q=\slim_{n\to\infty}Q_n
  \end{equation*}
  exists and is a projection.
\end{lemma}
\section{A commutator expansion formula}
\label{app:B}
\renewcommand{\ad}[2]{\adjungeret_{#1}^{#2}}
\setcounter{section}{2} 
\setcounter{theorem}{0}

In this section, we recall a result from \cite{MGR}.

In the following, $A=(A_1,\dotsc,A_\nu)$ is a vector of self-adjoint,
pairwise commuting operators acting on a Hilbert space $\cH$, and
$B\in\cB(\cH)$ is a bounded operator on $\cH$. We shall use the notion
of $B$ being of class $C^{n_0}(A)$ introduced in \cite{ABG}. For
notational convenience, we adopt the following convention: If $0\le
j\le\nu$, then $\kdj$ denotes the multi-index
$(0,\dotsc,0,1,0,\dotsc,0)$, where the $1$ is in the $j$'th entry.

\begin{definition}
  Let $n_0\in\N\cup\{\infty\}$. Assume that the multi-commutator form
  defined iteratively by $\ad{A}{0}(B)=B$ and
  $\ad{A}{\alpha}(B)=[\ad{A}{\alpha-\kdj}(B),A_j]$ as a form on
  $\cD(A_j)$, where $\alpha\ge\kdj$ is a multi-index and $1\le
  j\le\nu$, can be represented by a bounded operator also denoted by
  $\ad{A}{\alpha}(B)$, for all multi-indices $\alpha$,
  $\abs{\alpha}<n_0+1$. Then $B$ is said to be of class $C^{n_0}(A)$
  and we write $B\in C^{n_0}(A)$. 
\end{definition}

\begin{remark}
  The definition of $\ad{A}{\alpha}(B)$ does not depend on the order
  of the iteration since the $A_j$ are pairwise commuting. We call
  $\abs\alpha$ the \emph{degree} of $\ad{A}{\alpha}(B)$.
\end{remark}

In the following, $\cH_A^s := D(\abs{A}^s)$ for $s\ge0$ will be used
to denote the scale of spaces associated to $A$. For negative $s$, we
define $\cH_A^s := (\cH_A^{-s})^*$.

\begin{theorem}\label{thm:MGR}
  Assume that $B\in C^{n_0}(A)$ for some $n_0\ge n+1\ge 1$, $0\le
  t_1,t_2$, $t_1+t_2\le n+2$ and that $\{f_\lambda\}_{\lambda\in I}$
  satisfies 
  \begin{equation*}
    \label{eq:AAE}
    \forall\alpha\,\exists C_\alpha\colon\abs{\partial^\alpha
      f_\lambda(x)}\le C_\alpha\jnorm{x}^{s-\abs\alpha}
  \end{equation*}
  uniformly in $\lambda$ for some $s\in\R$ such that
  $t_1+t_2+s<n+1$. Then
  \begin{align*}
    [B,f_\lambda(A)]=\sum_{\abs\alpha=1}^n\frac{1}{\alpha!}\partial^\alpha
    f_\lambda(A)\,\ad{A}{\alpha}(B)+R_{\lambda,n}(A,B)
  \end{align*}
  as an identity on $\cD(\jnorm{A}^s)$, where
  $R_{\lambda,n}(A,B)\in\cB(\cH_A^{-t_2},\cH_A^{t_1})$ and there exist
  a constant $C$ independent of $A$, $B$ and $\lambda$ such that
  \begin{equation*}
    \norm{R_{\lambda,n}(A,B)}_{\cB(\cH_A^{-t_2},\cH_A^{t_1})}\le
    C\smashoperator{\sum_{\abs{\alpha}=n+1}}\norm{\ad{A}{\alpha}(B)}.
\end{equation*}
\end{theorem}

 \providecommand{\noopsort}[1]{}

\end{document}